\documentclass[12pt]{amsart}
\usepackage{amssymb,color}
\numberwithin{equation}{section} %\usepackage{amssymb}

\usepackage{graphicx}
\newcommand{\bea}{\begin{eqnarray}}
\newcommand{\eea}{\end{eqnarray}}
\newcommand{\ba}{\begin{array}}
\newcommand{\ea}{\end{array}}
\newcommand{\edc}{\end{document}}
\newcommand{\bc}{\begin{center}}
\newcommand{\ec}{\end{center}}
\newcommand{\be}{\begin{equation}}
\newcommand{\ee}{\end{equation}}

\newcommand{\dsf}{\displaystyle\frac}

\def\cb{{\mathcal B}}

\def\ce{{\mathcal E}}

\def\cg{{\mathcal G}}

\def\bc{{\mathbb C}}

\def\bn{{\mathbb N}}

\def\bq{{\mathbb Q}}
\def\br{{\mathbb R}}

\def\bz{{\mathbb Z}}

\def\g{\gamma}  \def\G{\Gamma}

\def\e{\epsilon}
\def\l{\lambda} 

\def\m{\mu}

\def\n{\nu}
\def\r{\rho}
\def\s{\sigma}

\def\w{\omega} \def\Om{\Omega}

\def\h{{\mathbf{h}}}

\def\sb{{\mathbf{s}}}

\newtheorem{thm}{Theorem}[section]
\newtheorem{lem}[thm]{Lemma}
\newtheorem{cor}[thm]{Corollary}

\newtheorem{prop}[thm]{Proposition}

\theoremstyle{remark}
\newtheorem{rem}{Remark}[section]

\begin{document}

\title[$p$-adic $\lambda$-model]
{Renormalization method in $p$-adic $\lambda$-model on the Cayley
tree}

%\thanks{On leave from Department of Mechanics and Mathematics,
%National University of Uzbekistan, Tashkent, 700174, Uzbekistan}

\author{Farrukh Mukhamedov}
\address{Farrukh Mukhamedov\\
 Department of Computational \& Theoretical Sciences\\
Faculty of Science, International Islamic University Malaysia\\
P.O. Box, 141, 25710, Kuantan\\
Pahang, Malaysia} \email{{\tt far75m@yandex.ru} {\tt farrukh\_m@iium.edu.my}}

\begin{abstract}
In this present paper, it is proposed the renormalization techniques
in the investigation of phase transition phenomena in $p$-adic
statistical mechanics. We mainly study $p$-adic $\l$-model on the
Cayley tree of order two. We consider generalized $p$-adic quasi
Gibbs measures depending on parameter $\r\in\bq_p$, for the
$\l$-model. Such measures are constructed by means of certain
recurrence equations. These equations define a dynamical system. We
study two regimes with respect to parameters. In the first regime we
establish that the dynamical system has one attractive and two
repelling fixed points, which predicts the existence of a phase
transition. In the second regime the system has two attractive and
one neutral fixed points, which predicts the existence of a quasi
phase transition. A main point of this paper is to verify (i.e.
rigorously prove)  and confirm that the indicated predictions (via
dynamical systems point of view) are indeed true.
 To establish the main result, we employ the methods of
$p$-adic analysis, and therefore, our results are not valid in the
real setting.

\vskip 0.3cm \noindent {\it
Mathematics Subject Classification}: 82B26, 46S10, 12J12, 39A70, 47H10, 60K35.\\
{\it Key words}: $p$-adic numbers, $\l$-model; $p$-adic quasi Gibbs
measure, strong phase transition, Cayley tree.
\end{abstract}

\maketitle

\section{introduction}

After WilsonТs seminal work in the early 1970's \cite{Wil}, based
also on the ground breaking foundations laid by Kadanoff, Widom,
Michael Fisher \cite{Fish}, and others in the preceding decade, the
renormalization group (RG) has had a profound impact on modern
statistical physics. Not only do renormalization group techniques
provide a powerful tool to analytically describe and quantitatively
capture both static and dynamic critical phenomena near continuous
phase transitions that are governed by strong interactions,
fluctuations, and correlations. RG presents a conceptual framework
and mathematical language that has become ubiquitous in the
theoretical description of many complex interacting many-particle
systems encountered in nature (see for review \cite{Ta}).

The renormalization method is then applied in statistical mechanics
and yielded lots of interesting results. Since such investigations
of phase transitions of spin models on hierarchical lattices showed
that they make the exact calculation of various physical quantities
\cite{Bax,G}. One of the simplest hierarchical lattice is Cayley
tree or Bethe lattice (see \cite{Ost}). This lattice is not a
realistic lattice, however, investigations of phase transitions of
spin models on trees like Cayley tree  showed that they make the
exact calculation of various physical quantities \cite{Roz}. It is
believed that several among its interesting thermal properties could
persist for regular lattices, for which the exact calculation is far
intractable. Illustrations of the renormalization methods are widely
shown in the study of Ising model \cite{Bax}, since it has wide
theoretical interest and practical applications. Therefore, one of
generalizations of the Ising model is so-called $\l$-model on the
Cayley tree (see \cite{R1,M04}). Such a model has enough rich
structure to illustrate almost every conceivable nuance of
statistical mechanics. We have to stress that one of the central
problems in the theory of Gibbs measures of lattice systems is to
describe infinite-volume (or limiting) Gibbs measures corresponding
to a given Hamiltonian. A complete analysis of this set is often a
difficult problem (see for review \cite{G,Roz}).

On the other hand, there are many investigates have been done to
discuss and debate the question due to the assumption that $p$-adic
numbers provide a more exact and more adequate description of
microworld phenomena (see for example \cite{DFR,Sn,V1}). Therefore,
starting the 1980s, various models described in the language of
$p$-adic analysis have been actively studied
\cite{ADFV},\cite{FO},\cite{MP}. The well-known studies in this area
are primarily devoted to investigating quantum mechanics models
using equations of mathematical physics
\cite{ABK,ADV,Kh1,Kh2,V2,VVZ}. We refer the reader to \cite{DKKV}
for recent development of the subject.

One of the first applications of $p$-adic numbers in quantum physics
appeared in the framework of quantum logic in \cite{BC}. This model
is especially interesting for us because it could not be described
by using conventional real valued probability (see
\cite{Kh2,Ko,MP,VVZ}). Therefore, $p$-adic probability models were
investigated in \cite{K3,KhN,KYR}. Using that $p$-adic measure
theory in \cite{Kh07,KL,Lu}, the theory of $p$-adic and
non-Archimedean stochastic processes has been developed. These
investigations allowed us to construct wide classes of stochastic
processes using finite dimensional probability distributions
\cite{GMR}. In \cite{GRR},\cite{Mq}-\cite{MDA},\cite{MR1,MR2} it has
been developed $p$-adic statistical mechanics within the scheme of
the theory of $p$-adic probability and $p$-adic stochastic
processes. Namely, we have studied $p$-adic Ising and Potts models
with nearest neighbor interactions on Cayley trees.

In the present paper, we propose to study phase transition phenomena
of $p$-adic statistical models by means of renormalization methods
in the measure-theoretical scheme. Note that the renormalization
method is closely related to the investigation of dynamical system
associated with a given model. Therefore, in what follows, methods
of $p$-adic dynamical systems and $p$-adic probability measures will
be used. In this paper, we illustrate our propose in the study of
$p$-adic $\l$-model which was started in \cite{KM,KMR}. In this
model spin takes two different values. In the mentioned papers we
studied only the uniqueness of $p$-adic Gibbs measures of the model.
Recently, in \cite{M12,M13} it was introduced two kind of notions of
phase transition: \textit{phase transition} and \textit{quasi phase
transition}. Note that the investigate of phase transitions by
dynamical system approach, in real case, has greatly enhanced
understanding of complex properties of models. The interplay of
statistical mechanics with chaos theory has even led to novel
conceptual frameworks in different physical settings \cite{E}.
Therefore, a main aim of this paper is to apply and verify
renormalization method to the existence of phase transitions.

Let us highlight the organization of the paper. In section 2 we
collect necessary definitions and preliminary results which will be
used in the paper. In section 3 we provide a measure-theoretical
construction of generalized $p$-adic quasi Gibbs measures for the
$\l$-model. Such kind of measures exist if the interacting functions
satisfy certain recurrence equation. In section 4 we consider two
regimes with respect to a parameters $A$ and $C$. In this section we
prove the existence of generalized $p$-adic Gibbs measures in both
regimes. The obtained recurrence equations define a dynamical
system. In the first regime we establish that the dynamical system
has one attractive and two repelling fixed points, which predicts
the existence of a phase transition. In the second regime the system
has two attractive and one neutral fixed points, which predicts the
existence of a quasi phase transition. In section 5, we verify (i.e.
rigorously prove)  and confirm that the indicated predictions (via
dynamical systems point of view) are indeed true.
 To establish the main result, we employ the methods of
$p$-adic analysis, and therefore, our results are not valid in the
real setting.

\section{Preliminaries}

\subsection{$p$-adic numbers}

In what follows $p$ will be a fixed prime number. The set $\bq_p$ is
defined as a completion of the rational numbers $\bq$ with respect
to the norm $|\cdot|_p:\bq\to\br$ given by
\begin{eqnarray}
|x|_p=\left\{
\begin{array}{c}
  p^{-r} \ x\neq 0,\\
  0,\ \quad x=0,
\end{array}
\right.
\end{eqnarray}
here, $x=p^r\frac{m}{n}$ with $r,m\in\bz,$ $n\in\bn$,
$(m,p)=(n,p)=1$. The absolute value $|\cdot|_p$ is non-Archimedean,
meaning that it satisfies the \textit{strong triangle inequality}
$|x + y|_p \leq \max\{|x|_p, |y|_p\}$. We recall a nice property of
the norm, i.e. if $|x|_p>|y|_p$ then $|x+y|_p=|x|_p$. Note that this
is a crucial property which is proper to the non-Archimedenity of
the norm.

Any $p$-adic number $x\in\bq_p$, $x\neq 0$ can be uniquely represented in the form
\begin{equation}\label{canonic}
x=p^{\g(x)}(x_0+x_1p+x_2p^2+...),
\end{equation}
where $\g=\g(x)\in\bz$ and $x_j$ are integers, $0\leq x_j\leq p-1$,
$x_0>0$, $j=0,1,2,\dots$ In this case $|x|_p=p^{-\g(x)}$.

We recall that an integer $a\in \bz$ is called {\it a quadratic
residue modulo $p$} if the equation $x^2\equiv a(\textrm{mod $p$})$
has a solution $x\in \bz$.

\begin{lem}\label{quadrat} \cite{Ko} In order that the equation
$$
x^2=a, \ \ 0\neq a=p^{\g(a)}(a_0+a_1p+...), \ \ 0\leq a_j\leq p-1, \
a_0>0
$$
has a solution $x\in \bq_p$, it is necessary and sufficient that the
following conditions are fulfilled:
\begin{enumerate}
\item[(i)] $\g(a)$ is even;\\
\item[(ii)] $a_0$ is a quadratic residue modulo $p$ if $p\neq 2$, and moreover
$a_1=a_2=0$ if $p=2$.
\end{enumerate}
\end{lem}

For each $a\in \bq_p$, $r>0$ we denote $$ B(a,r)=\{x\in \bq_p :
|x-a|_p< r\}, \ \ \bz_{p}=\left\{ x\in \bq_{p}:\
|x|_{p}\leq1\right\}.$$

Recall that the $p$-adic exponential is defined by
$$
\exp_p(x)=\sum_{n=1}^{\infty}\dsf{x^n}{n!},
$$
which converges for every $x\in B(0,p^{-1/(p-1)})$. It is known
\cite{Ko} that for any $x\in B(0,p^{-1/(p-1)})$ one has
$$ |\exp_p(x)|_p=1,\ \ \ |\exp_p(x)-1|_p=|x|_p<1. $$

Put
\begin{equation}\label{Exp}
\ce_p=\{x\in\bq_p: \ |x|_p=1, \ \ |x-1|_p<p^{-1/(p-1)}\}.
\end{equation}

Note that the basics of $p$-adic analysis, $p$-adic mathematical
physics are explained in \cite{Ko,S,VVZ}.

Now we recall some standard terminology of the theory of dynamical
systems (see for example \cite{AKh},\cite{KhN}).

Let $(f,B)$ be a dynamical system in $\bq_p$, where $f: x\in B\to
f(x)\in B$ is some function and $B=B(a,r))$ or $\bq_p$. Denote
$x^{(n)}=f^n(x^{(0)})$, where $x^0\in B$ and
$f^n(x)=\underbrace{f\circ\dots\circ f(x)}_n$.
 If $f(x^{(0)})=x^{(0)}$ then $x^{(0)}$
is called a {\it fixed point}. Let $x^{(0)}$ be a fixed point of an
analytic function $f(x)$. Set
$$
\l=\frac{d}{dx}f(x^{(0)}).
$$

The point $x^{(0)}$ is called {\it attractive} if $0\leq |\l|_p<1$,
{\it neutral} if $|\l|_p=1$, and {\it repelling} if $|\l|_p>1$.

It is known \cite{KhN} that if a fixed point $x^{(0)}$ is attractive
then there exists a neighborhood $U(x^{(0)})(\subset B)$ of
$x^{(0)}$ such that for all points $y\in U(x^{(0)})$ it holds
$\lim\limits_{n\to\infty}f^{(n)}(y)=x^{(0)}$. If a fixed point
$x^{(0)}$ is  repelling, then there  exists a neighborhood
$U(x^{(0)})$ of $x^{(0)}$ such that $|f(x)-x^{(0)}|_p>|x-x^{(0)}|_p$
for $x\in U(x^{(0)})$, $x\neq x^{(0)}$.

\subsection{$p$-adic measure}

Let $(X,\cb)$ be a measurable space, where $\cb$ is an algebra of
subsets $X$. A function $\m:\cb\to \bq_p$ is said to be a {\it
$p$-adic measure} if for any $A_1,\dots,A_n\subset\cb$ such that
$A_i\cap A_j=\emptyset$ ($i\neq j$) the equality holds
$$
\mu\bigg(\bigcup_{j=1}^{n} A_j\bigg)=\sum_{j=1}^{n}\mu(A_j).
$$

A $p$-adic measure is called a {\it probability measure} if
$\mu(X)=1$.  One of the important condition (which was already
invented in the first Monna--Springer theory of non-Archimedean
integration \cite{Mona}) is boundedness, namely a $p$-adic
probability measure $\m$ is called {\it bounded} if $\sup\{|\m(A)|_p
: A\in \cb\}<\infty $. We pay attention to an important special case
in which boundedness condition by itself provides a fruitful
integration theory (see for example \cite{Kh07}). Note that, in
general, a $p$-adic probability measure need not be bounded
\cite{Ro}. For more detail information about $p$-adic measures we
refer to \cite{K3,KhN,Ro}.

\subsection{Cayley tree}

Let $\Gamma^k_+ = (V,L)$ be a semi-infinite Cayley tree of order
$k\geq 1$ with the root $x^0$ (whose each vertex has exactly $k+1$
edges, except for the root $x^0$, which has $k$ edges). Here $V$ is
the set of vertices and $L$ is the set of edges. The vertices $x$
and $y$ are called {\it nearest neighbors} and they are denoted by
$l=<x,y>$ if there exists an edge connecting them. A collection of
the pairs $<x,x_1>,\dots,<x_{d-1},y>$ is called a {\it path} from
the point $x$ to the point $y$. The distance $d(x,y), x,y\in V$, on
the Cayley tree, is the length of the shortest path from $x$ to $y$.
$$
W_{n}=\left\{ x\in V\mid d(x,x^{0})=n\right\}, \ \
V_n=\overset{n}{\underset{m=0}{\bigcup}}W_{m}, \ \ L_{n}=\left\{
l=<x,y>\in L\mid x,y\in V_{n}\right\}.
$$
The set of direct successors of $x$ is defined by
$$
S(x)=\left\{ y\in W_{n+1}:d(x,y)=1\right\}, x\in W_{n}.
$$
Observe that any vertex $x\neq x^{0}$ has $k$ direct successors and
$x^{0}$ has $k+1$.

Given a set $A$, by $|A|$ we denote the number of its elements. In
what follows, we need the following equalities:
\begin{eqnarray}\label{WVn1}
&&|W_n|=k^n, \ \ |V_n|=\frac{k^{n+1}-1}{k-1}, \\[2mm]
\label{WVn2} && |W_n|=(k-1)|V_{n-1}|+1, \ \ |V_n|=k|V_{n-1}|+1.
\end{eqnarray}

\section{$p$-adic $\l$ model and its $p$-adic quasi Gibbs measures}

In this section we consider the $p$-adic $\l$-model where spin takes
values in the set $\Phi=\{-1,+1\}$, ($\Phi$ is called a {\it state
space}) and is assigned to the vertices of the tree $\G^k_+=(V,L)$.
A configuration $\s$ on $V$ is then defined as a function $x\in
V\to\s(x)\in\Phi$; in a similar manner one defines configurations
$\s_n$ and $\w$ on $V_n$ and $W_n$, respectively. The set of all
configurations on $V$ (resp. $V_n$, $W_n$) coincides with
$\Omega=\Phi^{V}$ (resp. $\Omega_{V_n}=\Phi^{V_n},\ \
\Omega_{W_n}=\Phi^{W_n}$). One can see that
$\Om_{V_n}=\Om_{V_{n-1}}\times\Om_{W_n}$. Using this, for given
configurations $\s_{n-1}\in\Om_{V_{n-1}}$ and $\w\in\Om_{W_{n}}$ we
define their concatenations  by
$$
(\s_{n-1}\vee\w)(x)= \left\{
\begin{array}{ll}
\s_{n-1}(x), \ \ \textrm{if} \ \  x\in V_{n-1},\\
\w(x), \ \ \ \ \ \ \textrm{if} \ \ x\in W_n.\\
\end{array}
\right.
$$
It is clear that $\s_{n-1}\vee\w\in \Om_{V_n}$.

Assume for each edge $<x,y>\in L$ a function $\l:\Phi\times\Phi\to
{\bz}$ is given. Then the Hamiltonian $H_n:\Om_{V_n}\to\bz$ of the
$p$-adic $\l$-model is defined by
\begin{equation}\label{ham1}
H_{n}(\sigma)=\sum\limits_{< x,y>\in L_{n}}\l(\sigma(x),\sigma(y))
\end{equation}

\begin{rem} This model first has been considered in \cite{KM}. In
the real setting such kind of model was studied in \cite{R1}. We
remark that if one takes $\l(u,v)=Nuv$ for some integer $N$, then
the model \eqref{ham1} reduces to the well-known Ising model (see
\cite{GMR,KM,Mq}).
\end{rem}

Let $\r\in\bq_p$ and assume that $\h:x\in {V\setminus\{x^{(0)}\}}\to
h_x\in\bq_p$ be a mapping. Given $n\in\bn$, let us consider a
$p$-adic probability measure $\m^{(n)}_\h$ on $\Om_{V_n}$ defined by

\begin{equation}\label{measure1}
\m_{\h,\r}^{(n)}(\sigma)=\frac{1}{Z_{n,\r}^{(\h)}}\r^{H_{n}(\sigma)}
{\underset{x\in W_{n}}{\prod}\left(h_{x}\right)^{\sigma(x)}}
\end{equation}
Here, $\s\in\Om_{V_n}$, and $Z_{n,\r}^{(\h)}$ is the corresponding
normalizing factor called a {\it partition function} given by

\begin{equation}\label{partition}
Z_{n,\r}^{(\h)}=\sum_{\s\in\Omega_{V_n}}\r^{H_{n}(\sigma)}{\underset{x\in
W_{n}}{\prod}\left(h_{x}\right)^{\sigma(x)}}.
\end{equation}

We recall \cite{M12} that one of the central results of the theory
of probability concerns a construction of an infinite volume
distribution with given finite-dimensional distributions, which is a
well-known {\it Kolmogorov's extension Theorem} \cite{Sh}. Recall
that a $p$-adic probability measure $\m$ on $\Om$ is
\textit{compatible} with defined ones $\m_\h^{(n)}$ if one has
\begin{equation}\label{CM}
\m(\s\in\Om: \s|_{V_n}=\s_n)=\m^{(n)}_{\h,\r}(\s_n), \ \ \
\textrm{for all} \ \ \s_n\in\Om_{V_n}, \ n\in\bn.
\end{equation}

The existence of the measure $\m$ is guaranteed by the $p$-adic
Kolmogorov's Theorem \cite{GMR,KL}. Namely, if the measures
$\m_{\h,\r}^{(n)}$, $n\geq 1$ satisfy the {\it compatibility
condition}, i.e.
\begin{equation}\label{comp}
\sum_{\w\in\Om_{W_n}}\m^{(n)}_{\h,\r}(\s_{n-1}\vee\w)=\m^{(n-1)}_{\h,\r}(\s_{n-1}),
\end{equation}
for any $\s_{n-1}\in\Om_{V_{n-1}}$, then there is a unique measure
$\m$ on $\Om$ with \eqref{CM}.

Now following \cite{M12} if for some function $\h$ the measures
$\m_{\h,\r}^{(n)}$ satisfy the compatibility condition, then there
is a unique $p$-adic probability measure, which we denote by
$\m_{\h,\r}$, since it depends on $\h$ and $\r$. Such a measure
$\m_\h$ is said to be {\it a generalized $p$-adic quasi Gibbs
measure} corresponding to the $p$-adic $\l$-model. By $Q\cg(H)$ we
denote the set of all generalized $p$-adic quasi Gibbs measures
associated with functions $\h=\{\h_x,\ x\in V\}$. If there are at
least two distinct generalized $p$-adic quasi Gibbs measures
$\m,\n\in Q\cg(H)$ such that $\m$ is bounded and $\n$ is unbounded,
then we say that {\it a phase transition} occurs. By another words,
one can find two different functions $\sb$ and $\h$ defined on $\bn$
such that there exist the corresponding measures $\m_{\sb,\r}$ and
$\m_{\h,\r}$, for which one is bounded, another one is unbounded.
Moreover, if there is a sequence of sets $\{A_n\}$ such that
$A_n\in\Om_{V_n}$ with $|\m(A_n)|_p\to 0$ and $|\n(A_n)|_p\to\infty$
as $n\to\infty$, then we say that there occurs a {\it strong phase
transition}. If there are two different functions $\sb$ and $\h$
defined on $\bn$ such that there exist the corresponding measures
$\m_{\sb,\rho}$, $\m_{\h,\rho}$, and they are bounded, then we say
there is a \textit{quasi phase transition}.

Note that some comparison of these phase transitions with real
counterparts was highlighted in \cite{M12}. In \cite{M12,MA13} the
existence of the strong phase transition for the $q+1$-state Potts
model on the Caylay tree has been proved. In the present paper, we
are going to establish such kind of phenomena for the $\l$-model.

One can prove the following theorem.

\begin{thm}\label{compatibility}
The measures $\m^{(n)}_{\h,\r}$, $ n=1,2,\dots$ (see
\eqref{measure1}), associated with $\l$-model \eqref{ham1}, satisfy
the compatibility condition \eqref{comp} if and only if for any
$x\in V\setminus\{x^{(0)}\}$ the following equation holds:

\begin{equation}\label{canonic3}
h_x^2=\prod_{y\in S(x)}\left(\frac{\r^{\l(1,1)}h_y^2+\r^{\l(1,-1)}}
{\r^{\l(-1,1)}h_y^2+\r^{\l(-1,-1)}}\right).
\end{equation}
\end{thm}

The proof can be proceeded by the same argument as in \cite{KM}.

According to Theorem \ref{compatibility} the problem of describing
the generalized $p$-adic quasi Gibbs measures is reduced to the
description of solutions of the functional equations
\eqref{canonic3}.

\section{Dynamical system and the existence of generalized $p$-Adic quasi Gibbs Measures}

In this section we consider the $\l$-model \eqref{ham1} over the
Cayley tree of order two, i.e. $k=2$. Main aim of this section is to
establish the existence of generalized $p$-adic quasi Gibbs measures
by analyzing the equation \eqref{canonic3}. In the sequel, we will
consider a case when $|\r|_p<1$ and $p\geq 3$. Note that the case
$\r\in\ce_p$ has been studied in \cite{KM,MD}.

Recall that a function $\h=\{\h_x\}_{x\in V\setminus\{x^0\}}$ is
called {\it translation-invariant} if $\h_{x}=\h_{y}$ for all
$x,y\in V$. A $p$-adic measure $\m_\h$, corresponding to a
translation-invariant function $\h$, is called a {\it
translation-invariant generalized $p$-adic quasi Gibbs measure}.

To solve the equation \eqref{canonic3}, in general, is very
complicated. Therefore,
 let us first restrict ourselves to
the description of translation-invariant solutions of
\eqref{canonic3}. More exactly, we suppose that $h_{x}:=h$ for all
$x\in V$. Then from \eqref{canonic3} we find
\begin{equation}\label{nnn}
h^2=\left(\frac{Ah^2+B} {Ch^2+D}\right)^2,
\end{equation}
where $A=\r^{\l(1,1)}$,$B=\r^{\l(1,-1)}$, $C=\r^{\l(-1,1)}$,
$D=\r^{\l(-1,-1)}$.

The last equation splits into the following ones:
\begin{equation}\label{nnn31}
h=\bigg(\frac{Ah^2+B} {Ch^2+D}\bigg),
\end{equation}

\begin{equation}\label{nnn32}
h=\bigg(\frac{Ah^2+B} {Ch^2+D}\bigg).
\end{equation}

One can see that \eqref{nnn32} is conjugate to \eqref{nnn31} via
$h(x)=-x$. Therefore, we will investigate the equation
\eqref{nnn31}.

In this paper, we restrict ourselves to a special case. Namely, we
assume that $|A|_p,|C|_p<1$, $B=D=1$, i.e. $\l(1,1),\l(-1,1)\in\bn$,
$\l(1,-1)=\l(-1,-1)=0$. In what follows, we will assume that
$|A|_p\neq|C|_p$, otherwise one finds $A=C$ and correspondingly
equation \eqref{nnn31} becomes trivial.

Let us denote
$$
S=\{x\in\bq_p:\ \ |x|_p=1\}.
$$

\begin{lem}\label{parti-g1}
Let $p\geq3$, and $|A|_p,|C|_p<1$ and $f$ be given by
\begin{equation}\label{g11}
f(x)=\frac{Ax^2+1}{Cx^2+1}.
\end{equation}
Then $f(S)\subset S$ and
$$|f(x)-f(y)|_{p}\leq|A-C||x-y|_{p},$$
 for all $x,y\in S$.
\end{lem}

\begin{proof}
Assume that $u\in S$. Then from
\begin{equation}\label{g111}
|Au^2+1|_p=|Cu^2+1|_p=1
\end{equation}
one gets $f(S)\subset S$. Now let us show the second condition. Let
$x,y\in S$, then from \eqref{g111} we have
\begin{eqnarray*}
|f(x)-f(y)|_{p}&=& \left|\frac{C(y^2-x^2)+A(x^2-y^2)}
{(Cx^{2}+1)(Cy^{2}+1)}\right|_{p}\\[2mm]
&\leq&|A-C|_p|x-y|_{p}.
\end{eqnarray*}
This completes the proof.
\end{proof}

Now we can formulate the following proposition about fixed points of
$f$.

\begin{thm}\label{g1-fix} Let $|A|_p,|C|_p<1$ with $|A|_p\neq |C|_p$, and $f$
be given by \eqref{g11}. Then the following statements hold:
\begin{enumerate}
\item[(i)] The function $f$ has a unique fixed point $x_{0}$ in
$\ce_{p}$;
\item[(ii)] Assume that $|A|_p^2<|C|_p$. Then the function $f$ has at most two fixed points
$x_{1},x_{2}$ different from $x_0$ if and only if $\sqrt{-C}$
exists; Moreover, one has
\begin{equation}\label{x12-1}
|x_{1,2}|_p=\frac{1}{\sqrt{|C|_p}};
\end{equation}

\item[(iii)] Assume that $|A|_p^2>|C|_p$. Then the function $f$ has two fixed points
$x_{1},x_{2}$ different from $x_0$. Moreover, one has
\begin{equation}\label{x12-2}
|x_{1}|_p=\frac{|A|_p}{|C|_p}, \ |x_{2}|_p=\frac{1}{|A|_p}.
\end{equation}
\end{enumerate}
\end{thm}

\begin{proof} (i) By Lemma \ref{parti-g1} we conclude that $f$ satisfies the Banach
contraction principle on $S$. Therefore, there exists $x_0\in S$
such that $f(x_0)=x_0.$ Let us show that $x_0\in\ce_p$. Indeed, we
have
\begin{eqnarray}\label{g-x0}
|x_0-1|_p&=&|f(x_0)-1|_p=\bigg|\frac{Ax^2+1}{Cx^2+1}-1\bigg|_p\nonumber\\[2mm]
&=&|(A-C)x_0^2|_p\nonumber\\
&=&|A-C|_p<1
\end{eqnarray}
this means $x_0\in\ce_p$.\\

First note that the  equation $x=f(x)$ can be rewritten as follows
\[
Cx^{3}-Ax^{2}+x-1=0
\]

Note that, in general, we may solve the last equation by methods
developed in \cite{MOS,MOSM}. But those methods give only
information about the existence of solutions. In reality, we need
more properties of the solutions (see further sections). Therefore,
we are going to find all the solutions.

 Since $x_0$ is a solutions
of the last equation, therefore, one has

\begin{equation}\label{canonicU}
Cx^{3}-Ax^{2}+x-1=(x-x_{0})(Cx^{2}+(Cx_{0}-A)x+1/x_{0})).
\end{equation}

Let us solve
\begin{equation}\label{quadratic2}
Cx^{2}+(Cx_{0}-A)x+1/x_{0}=0.
\end{equation}

From (i) and the conditions of the proposition we can write
\begin{equation}\label{xAB}
\frac{1}{x_{0}}=1+\e_0p^{\g_0}, \ \ C=\e_1p^{\g_1}, \ \
A=\e_2p^{\g_2},
\end{equation}
where $\e_0,\e_2,\e_3\in\bz_p$ and $\g_0,\g_1,\g_1>0$.

(ii) From $|A|^2_p<|C|_p$ it follows that $2\g_2>\g_1$.

Hence, the discriminant of \eqref{quadratic2} can be calculated as
follows
\begin{eqnarray}\label{g-dis}
\Delta&=&(Cx_0-A)^2-\frac{4C}{x_0}\nonumber\\[2mm]
&=&p^{\g_1}\bigg(-\frac{4\e_1}{x_0}+\e_1^2x_0p^{\g_1}-2\e_1\e_2x_0p^{\g_2}+\e_2^2p^{2\g_2-\g_1}\bigg)
\nonumber\\[2mm]
&=&p^{\g_1}\bigg(-4\e_1-4\e_1\e_0p^{\g_0}+\e_1^2x_0p^{\g_1}-2\e_1\e_2x_0p^{\g_2}+\e_2^2p^{2\g_2-\g_1}\bigg)
\nonumber\\[2mm]
&=&p^{\g_1}(-4\e_1+\tilde\e p^{\tilde\gamma})
\end{eqnarray}
for some $\tilde\e\in\bz_p$ and $\tilde\g>0$.

Consequently, from Lemma \ref{quadrat} we conclude that
$\sqrt{\Delta}$ exists if and only if $\sqrt{-4C}$ exists, which is
equivalent the existence of $\sqrt{-C}$. So, it follows from
\eqref{g-dis} that $|\sqrt{\Delta}|_p=\sqrt{|C|_p}$.

Assume that \eqref{quadratic2} has two solutions $x_{1},x_{2}$,
which have the following form
\begin{eqnarray}\label{solx1}
x_{1,2}=\frac{A-Cx_{0}\pm\sqrt{\Delta}}{2C}.
\end{eqnarray}

Taking into account that
\begin{eqnarray}\label{ACC}
|A-C|_p\leq\max\{|A|_p,|C|_p\}<\sqrt{|C|_p},
\end{eqnarray}
 and
$|C(x_0+1)|_p=|C|_p$ with the strong triangle inequality from
 \eqref{solx1} we obtain
\begin{eqnarray*}\label{canonic61}
|x_{1,2}-1|_p&=&\frac{1}{|C|_p}|A-C-C(x_0+1)\pm\sqrt{\Delta}|_p\nonumber\\[2mm]
&=&\frac{1}{\sqrt{|C|_p}}.
\end{eqnarray*}

The last equality implies \eqref{x12-1}.

(iii) Now assume that $|A|_p^2>|C|_p$. This means that $2\g_2<\g_1$.
Then  from \eqref{g-dis} we find that
\begin{eqnarray}\label{g-dis}
\Delta=A^2(1+\delta_1 p^{\tilde\gamma_1})
\end{eqnarray}
for some $\delta_1\in\bz_p$ and $\tilde\g_1>0$. Hence, again from
Lemma \ref{quadrat} we infer that $\sqrt{\Delta}$ exists. Moreover,
one has $\sqrt{\Delta}=A(1+\delta_2 p^{\tilde\gamma_2})$, for some
$\delta_2\in\bz_p$ and $\tilde\g_2>0$. This yields that
$$
|A+\sqrt{\Delta}|_p=|A|_p, \ \ |A-\sqrt{\Delta}|_p<|A|_p.
$$

For the solutions $x_{1,2}$ (see \eqref{solx1})  from the last
equalities we obtain
\begin{eqnarray}\label{1x11}
|x_1|_p&=&\bigg|\frac{A+\sqrt{\Delta}-Cx_{0}}{2C}\bigg|_p=\frac{|A|_p}{|C|_p},
\end{eqnarray}
since $|A|_p>|A|_p^2>|C|_p$.

From the equality $x_1x_2=\frac{1}{Cx_0}$ with \eqref{1x11} one
finds
$$
|x_2|_p=\frac{1}{|A|_p}.
$$
This completes the proof.
\end{proof}

According to Theorem \ref{compatibility} the solutions $x_{0}$,
$x_{1}$ and $x_{2}$ (If they exist) generate generalized $p-$adic
quasi Gibbs measures $\mu_{0}$, $\mu_{1}$ and $\mu_{2}$,
respectively. Hence, we can formulate the following result.

\begin{thm}\label{exist-11}
Let $p\geq3$, $|\r|_p<1$. Assume that for the function $\l$ one has
\begin{equation}\label{lll}
 \l(1,1),\l(-1,1)>0,  \ \ \l(1,-1)=\l(-1,-1)=0.
\end{equation}
 Then for the $\l$-model \eqref{ham1} on the Cayley tree of order two the
following assertions hold:
\begin{enumerate}
\item[(i)] there exists a transition-invariant generalized $p$-adic quasi
Gibbs Measure $\mu_{0}$;

\item[(ii)] if
$$2\l(1,1)>\l(-1,1),
$$ then there are three
transition-invariant generalized $p$-adic quasi Gibbs measures
$\mu_{0}$, $\mu_{1}$ and $\mu_{2}$ if and only if
$\sqrt{-\r^{\l(-1,1)}}$ exists;

\item[(ii)] if
$$2\l(1,1)<\l(-1,1),$$ then there are three
transition-invariant generalized $p$-adic quasi Gibbs measures
$\mu_{0}$, $\mu_{1}$ and $\mu_{2}$.
\end{enumerate}
\end{thm}

In this paper, our main aim to establish the existence of phase
transitions for the model. In \cite{M13} we have proposed to predict
the phase transitions by looking at behavior of the function $f$.
Now we are going to determine behaviors of the fixed points of the
function.

\begin{prop}\label{dyn-bh}
Let $|A|_p,|C|_p<1$ with $|A|_p\neq |C|_p$, and $f$ be given by
\eqref{g11}. Then the following statements hold:
\begin{enumerate}
\item[(i)] The fixed point $x_0$ is attractive;\\
\item[(ii)] Assume that $|A|_p^2<|C|_p$ and $\sqrt{-C}$ exists.
Then the fixed points $x_{1,2}$ are repelling;

\item[(iii)] Assume that $|A|_p^2>|C|_p$. Then the fixed point
$x_{1}$ is attractive and $x_2$ is neutral.
\end{enumerate}
\end{prop}

\begin{proof} From \eqref{g11} we find that
\begin{equation}\label{der}
f'(x)=\frac{2(A-C)x}{(Cx^2+1)^2}.
\end{equation}

(i). Since $x_0\in\ce_p$, from \eqref{der} we get
$|f'(x_0)|_p=|A-C|_p<1$, which means that $x_0$ is attractive.

(ii). Assume that $|A|_p^2<|C|_p$ and $\sqrt{-C}$ exists. Then from
Theorem \ref{g1-fix} we conclude that the fixed points $x_{1,2}$
exist and satisfy the following equality
$$
Cx^{2}_{1,2}=(A-Cx_{0})x_{1,2}-\frac{1}{x_{0}}.
$$
Therefore, we have
\begin{eqnarray}\label{1x111}
|Cx_{1,2}^2+1|_p&=&\bigg|(A-Cx_{0})x_{1,2}-\frac{1}{x_{0}}+1\bigg|_p\nonumber\\[2mm]
&=&\bigg|(A-C)x_{1,2}-C(x_0-1)x_{1,2}+\frac{1-x_0}{x_0}\bigg|_p.
\end{eqnarray}

From \eqref{x12-1} and \eqref{g-x0} it follows that
$$
|(A-C)x_{1,2}|_p=\frac{|A-C|_p}{\sqrt{|C|_p}}, \ \
|C(x_0-1)x_{1,2}|_p=\sqrt{|C|_p}|A-C|_p, \ \ |1-x_0|_p=|A-C|_p.
$$
Hence,
\begin{equation*}\label{x120}
|(A-C)x_{1,2}|_p>|1-x_0|_p>|C(x_0-1)x_{1,2}|.
\end{equation*}
So, the last inequalities together with the strong triangle
inequality imply that \eqref{1x111} can be calculated as follows
\begin{eqnarray}\label{2x111}
|Cx_{1,2}^2+1|_p=\frac{|A-C|_p}{\sqrt{|C|_p}}.
\end{eqnarray}

Now from \eqref{der} with \eqref{2x111},\eqref{x12-1} one gets
\begin{eqnarray}\label{der1}
|f'(x_{1,2})|_p=\frac{|A-C|_p|x_{1,2}|_p}{|Cx_{1,2}^2+1|^2_p}=
\frac{\sqrt{|C|_p}}{|A-C|_p}>1
\end{eqnarray}
this implies that $x_{1,2}$ is repelling.

(iii). Assume that $|A|_p^2>|C|_p$, then the fixed points $x_{1,2}$
exist. Then from \eqref{x12-2} we immediately find
\begin{eqnarray}\label{3x111}
|Cx_{1}^2+1|_p=\frac{|A|_p^2}{|C|_p}, \ \ \
|Cx_{2}^2+1|_p=1.
\end{eqnarray}

Therefore, from \eqref{der}, \eqref{3x111} one gets
\begin{eqnarray}\label{der2}
|f'(x_{1})|_p=\frac{|C|_p}{|A|^2_p}<1, \ \ \ |f'(x_{2})|_p=1.
\end{eqnarray}
This means that $x_1$ is attractive and $x_2$ is neutral. The proof
is complete.
\end{proof}

\section{Phase Transitions}

In this section, we are going to establish the existence of the
phase transition for $\l$-models in the considered two regimes.

According to dynamical approach, taking into account Proposition
\ref{dyn-bh} we may predict that if $|A|_p^2<|C|_p$, then there
occurs a phase transition, and if $|A|_p^2>|C|_p$, then there exists
a quasi phase transition. In this section, we will confirm that our
predictions are true.

Before, going to prove main results we need some auxiliary facts.

\begin{lem}\label{AZ}\cite{MD} Ler $\r\in\bq_p$ and $\h$ be a solution of \eqref{canonic3}, and
$\m_{\h,\r}$ be an associated generalized $p$-adic quasi Gibbs
measure. Then for the corresponding partition function
$Z^{(\h)}_{n,\r}$ (see \eqref{partition}) the following equality
holds
\begin{equation}\label{ZN2}
|Z^{(\h)}_{n+1,\r}|_p=|A_{\h,n}|_p|Z^{(\h)}_{n,\r}|_p,
\end{equation}
where
\begin{equation}\label{aN3}
|A_{\h,n}|_p=\prod_{x\in W_n}|a(x)|_p,
\end{equation}
here
\begin{equation}\label{aN1}
|a(x)|_p^2=\bigg|\prod_{y\in S(x)}\sum_{\eta(y)\in\{-1,1\}}
\r^{\l(1,\eta(y))}(h_y)^{\eta(y)}\bigg|_p \bigg|\prod_{y\in
S(x)}\sum_{\eta(y)\in\{-1,1\}}
\r^{\l(-1,\eta(y))}(h_y)^{\eta(y)}\bigg|_p
\end{equation}
\end{lem}

From this lemma we immediately find the following

\begin{lem}\label{Znn}
Let $\h=\{h_x\}$ be a translation-invariant solution of
\eqref{canonic3}, i.e. $h_x=h_*$ for all $x\in V$. Then one has
\begin{equation}\label{ZNN}
Z^{(\h)}_{n,\r}=\frac{1}{|h_*|_p^{|V_{n-1}|}}\bigg|\r^{\l(-1,1)}h^2_*+\r^{\l(-1,-1)}\bigg|_p^{2|V_{n-1}|}.
\end{equation}
\end{lem}

\subsection{Regime $|A|_p^2<|C|_p$}

In this subsection our main result is the following result.

\begin{thm}\label{SPT}
Let $p\geq3$, $|\r|_p<1$. Assume that for the function $\l$ one has
\begin{equation*}\label{lll}
 2\l(1,1)>\l(-1,1),  \ \ \l(1,-1)=\l(-1,-1)=0.
\end{equation*}
and $\sqrt{-\r^{\l(-1,1)}}$ exists.
 Then there exist the phase transition for the
$\l$-model \eqref{ham1} on the Cayley tree of order two.
\end{thm}

\begin{proof} First we note that due to Theorem
\ref{exist-11} (ii) there are three translation-invariant
generalized $p$-adic Gibbs measures $\mu_0,\mu_1,\mu_2$.

Assume that $\h=\{h_x\}$ is a translation-invariant solution of
\eqref{canonic3}. Then $h_x=h_*$ for all $x\in V$, where $h_*$ is a
fixed point of $f$. Then due to Lemma \ref{Znn} from
\eqref{measure1} together with \eqref{ZNN} we obtain
\begin{eqnarray}\label{1mmm}
\left|\mu_{n,\r,*}(\sigma)\right|_{p}=\frac{|\r|_p^{H_n(\s)}|h_*|_p^{\sum_{x\in
W_n}\s(x)}|h_*|_p^{|V_{n-1}|}}{{\big|Ch^2_*+1\big|_p^{2|V_{n-1}|}}}.
\end{eqnarray}

Let us consider the measure $\m_0$. Since $x_0\in\ce_p$ (see
Proposition \ref{g1-fix}) and $|C|_p<1$, from \eqref{1mmm} one gets
\begin{eqnarray}\label{2mmm}
\left|\mu_{n,0}(\sigma)\right|_{p}=|\r|_p^{H_n(\s)}<1
\end{eqnarray}
This means that $\mu_{0}$ is bounded.\\

Now consider the measure $\m_{1,2}$ From \eqref{1mmm} together with
\eqref{2x111},\eqref{x12-1} one finds
\begin{eqnarray}\label{3mmm}
|\mu_{n,\r,1,2}(\sigma)|_{p}=\frac{|\r|_p^{H_n(\s)}|\sqrt{|C|_p}^{-\sum_{x\in
W_n}\s(x)}|\sqrt{|C|_p}^{|V_{n-1}|}}{|A-C|_p^{2|V_{n-1}|}}
\end{eqnarray}

Define a configuration $\s^{(-)}$ on $V_n$ by
$$
\s^{(-)}(x)=-1, \ \ \forall x\in V_n.
$$

Then one can see that $H_n(\s^{(-)})=0$.

Hence, from \eqref{3mmm} together with \eqref{WVn2},\eqref{ACC} we
have
\begin{eqnarray}\label{4mmm}
\left|\mu_{n,\r,1,2}(\s^{(-)})\right|_{p}&=&
\frac{\sqrt{|C|_p}^{|W_n|}\sqrt{|C|_p}^{|V_{n-1}|}}{|A-C|_p^{2|V_{n-1}|}}\nonumber\\[2mm]
&=&\sqrt{|C|_p}\bigg(\frac{\sqrt{|C|_p}}{|A-C|_p}\bigg)^{2|V_{n-1}|}\to
\infty \ \ \ \textrm{as} \ \ n\to\infty.
\end{eqnarray}
which implies $\mu_{1,2}$ is unbounded.

Hence, \eqref{2mmm} and \eqref{4mmm} imply the existence of the
phase transition. This completes the proof.
\end{proof}

\begin{rem} This proved theorem confirms that if the dynamical system associated with a model
 has at least
two repelling fixed points, then for the model exhibits a phase
transition. We stress that the considered $\l$-model has the
stronger phase transition (see \cite{MD}).
\end{rem}

\begin{rem} If one takes $\r=p$ and $\l(-1,1)=2m$ for some
$m\in\bn$, then $\sqrt{-p^{2m}}$ exists if and only if $p\equiv
1(\textrm{mod}\ 4)$.
\end{rem}

\subsection{Regime $|A|_p^2>|C|_p$}

In this subsection we prove the following result.

\begin{thm}\label{QPT}
Let $p\geq3$, $|\r|_p<1$. Assume that for the function $\l$ one has
\begin{equation}\label{lll-1}
 2\l(1,1)<\l(-1,1),  \ \ \l(1,-1)=\l(-1,-1)=0.
\end{equation}
 Then there exist the quasi phase transition for the
$\l$-model \eqref{ham1} on the Cayley tree of order two.
\end{thm}

\begin{proof} Theorem
\ref{exist-11} (iii) implies the existence of three
translation-invariant generalized $p$-adic Gibbs measures
$\mu_0,\mu_1,\mu_2$.

We remark that according to Proposition \ref{dyn-bh} the dynamical
system $f$ has two attractive and one neutral fixed points. This
indicated to the existence of quasi phase transition.

In the considered regime, by the same argument as in the proof of
Theorem \ref{SPT} one can establish that the measure $\m_0$ is also
bounded. Moreover, one can find
\begin{eqnarray}\label{3mmm1}
|\mu_{n,\r,1,2}(\sigma)|_{p}&=&\frac{|\r|_p^{H_n(\s)}|x_{1,2}|_p^{\sum_{x\in
W_n}\s(x)}|x_{1,2}|_p^{|V_{n-1}|}}{{\big|Cx_{1,2}^2+1\big|_p^{2|V_{n-1}|}}}.
\end{eqnarray}

Now consider the measure $\mu_{n,\r,1}$. Then from \eqref{3mmm1},
\eqref{x12-2} and \eqref{3x111} we obtain
\begin{eqnarray}\label{4mmm1}
|\mu_{n,\r,1}(\sigma)|_{p}&=&\frac{|\r|_p^{H_n(\s)}\bigg(\frac{|A|_p}{|C|_p}\bigg)^{\sum_{x\in
W_n}\s(x)+|V_{n-1}|}}{\bigg(\frac{|A|_p^2}{|C|_p}\bigg)^{2|V_{n-1}|}}\nonumber
\\[2mm]
&=&\frac{|\r|_p^{H_n(\s)}|A|_p^{\sum_{x\in
W_n}\s(x)-3|V_{n-1}|}}{|C|_p^{\sum_{x\in W_n}\s(x)-|V_{n-1}|}}.
\end{eqnarray}

Due to $H_n(\s)\leq |V_n|-1$ and $|C|_p<|A|_p$ we find
\begin{eqnarray}\label{4mmm2}
|\r|^{H_n(\s)}\leq |A|_p^{|V_n|-1}.
\end{eqnarray}
Taking into account the last expression with
$$-|W_n|\leq \sum_{x\in
W_n}\s(x)\leq |W_n|,
$$
and \eqref{WVn2} from \eqref{4mmm1} one gets
\begin{eqnarray*}
|\mu_{n,\r,1}(\sigma)|_{p}&\leq&\frac{|A|_p^{|V_n|-3|V_{n-1}|-1+\sum_{x\in
W_n}\s(x)}}{|C|_p^{\sum_{x\in
W_n}\s(x)-|V_{n-1}|}}\nonumber\\[2mm]
&=&\bigg(\frac{|C|_p}{|A|_p}\bigg)^{|V_{n-1}|-\sum_{x\in
W_n}\s(x)}\nonumber\\[2mm]
&\leq&\bigg(\frac{|C|_p}{|A|_p}\bigg)^{|V_{n-1}|-|W_n|}\nonumber\\[2mm]
&=&\frac{|C|_p}{|A|_p}<1.
\end{eqnarray*}
This means that the measure $\m_1$ is bounded.\\

Let us consider the measure $\mu_{n,\r,2}$. Then from \eqref{3mmm1},
\eqref{x12-2}, \eqref{3x111},\eqref{4mmm2} one finds
\begin{eqnarray}\label{5mmm1}
|\mu_{n,\r,2}(\sigma)|_{p}&=&|\r|_p^{H_n(\s)}\bigg(\frac{1}{|A|_p}\bigg)^{\sum_{x\in
W_n}\s(x)+|V_{n-1}|}\\[2mm]
&\leq &|A|_p^{|V_n|-1-\sum_{x\in W_n}\s(x)-|V_{n-1}|}\nonumber
\\[2mm]
&=&|A|_p^{|V_{n-1}|-\sum_{x\in W_n}\s(x)}\nonumber
\\[2mm]
&\leq&|A|_p^{|V_{n-1}|-|W_{n-1}|}\nonumber
\\[2mm]
&=&\frac{1}{|A|_p}\nonumber
\end{eqnarray}
This means that the measure $\m_2$ is bounded as well.

Consequently, we infer the existence of the quasi phase transition.
This completes the proof.
\end{proof}

By $\s\lceil_{W_n}$ we denote the restriction of a configuration
$\s$ to $W_n$. Define a configurations $\s^{(\pm)}_{n}$ on $W_n$ by
$$
\s^{(\pm)}_n(x)=\pm 1, \ \ \forall x\in W_n.
$$

\begin{cor} Let $p\geq3$, $|\r|_p<1$ and assume \eqref{lll-1} is satisfied.
Let
\begin{eqnarray}\label{Ann}
A_n^{(\pm)}=\{\s\in\Omega_{V_n}: \ \s\lceil_{W_n}=\s^{(\pm)}_n\}
\end{eqnarray}
 Then one has
\begin{eqnarray}\label{m12-1}
&&\bigg|\frac{\m_{n,\r,1}(\s)}{\m_{n,\r,2}(\s)}\bigg|_p\to 0,
\ \ \ n\to\infty, \ \ \textrm{for all} \ \ \s\in A_n^{(-)}\\[2mm]
\label{m12-2}
&&\bigg|\frac{\m_{n,\r,1}(\s)}{\m_{n,\r,2}(\s)}\bigg|_p\to \infty, \
\ \ \ n\to\infty,  \ \ \textrm{for all} \ \ \s\in A_n^{(+)}.
\end{eqnarray}
\end{cor}

\begin{proof} Take any $\s$ from $A_n^{(-)}$. Then
$$\sum_{x\in W_n}\s(x)=-|W_n|,$$ so from \eqref{4mmm1},\eqref{5mmm1}
we get
\begin{eqnarray}\label{m12-3}
\bigg|\frac{\m_{n,\r,1}(\s)}{\m_{n,\r,2}(\s)}\bigg|_p
&=&\frac{|A|_p^{\sum_{x\in W_n}\s(x)-2|V_{n-1}|}}{|C|_p^{\sum_{x\in
W_n}\s(x)-|V_{n-1}|}}\nonumber\\[2mm]
&=&\frac{|A|_p^{-|W_n|-2|V_{n-1}|}}{|C|_p^{-|W_n|-|V_{n-1}|}}\nonumber\\[2mm]
&=&\frac{|C|_p^{2|V_{n-1}|+1}}{|A|_p^{3|V_{n-1}|+1}}\nonumber\\[2mm]
&=&\frac{|C|_p}{|A|_p}\bigg(\frac{|C|_p^2}{|A|_p^3}\bigg)^{|V_{n-1}|}.
\end{eqnarray}

From $|C|_p<|A|_p^2$ we infer that $|C|^2_p<|A|^4_p<|A|_p^3$, which
with \eqref{m12-3} implies \eqref{m12-1}.

Now take $\s\in A_n^{(+)}$. Then using the same argument as above
one finds
\begin{eqnarray}\label{m12-4}
\bigg|\frac{\m_{n,\r,1}(\s)}{\m_{n,\r,2}(\s)}\bigg|_p
&=&\frac{|A|_p^{|W_n|-2|V_{n-1}|}}{|C|_p^{|W_n|-|V_{n-1}|}}\nonumber\\[2mm]
&=&\frac{|A|_p^{1-|V_{n-1}|}}{|C|_p}\nonumber\\[2mm]
&\geq&\frac{|A|_p}{|A|_p^{|V_{n-1}|}}\to \infty \ \ \textrm{as} \ \
n\to\infty.
\end{eqnarray}
This completes the proof.
\end{proof}

\begin{rem} This corollary shows that the bounded measures $\m_1$ and $\m_2$ are "singular"
on the sets $A_n^(\pm)$, which yields that they are different from
each other.
\end{rem}

\section{Conclusions}

It is known that in the investigate of phase transitions the
renomalization method is one of the powerful tools in theoretical
and mathematical physics. In real case, this method has greatly
enhanced our understanding of complex properties of models. The
interplay of statistical mechanics with chaos theory has even led to
novel conceptual frameworks in different physical settings \cite{E}.
Therefore, in the present paper, we have proposed to investigate
phase transition phenomena from renomalization technique
perspective. In the paper, we considered $p$-adic $\l$-model on the
Cayley tree. Note that if one takes $\l(x,y)=Nxy$, then such a model
reduces to the Ising model. This model was studied in \cite{MD,M13}.
But in the paper, we have concentrated ourselves to a totally
different model than the Ising one. For such a model, we have
considered two regimes with respect to a parameters $A$ and $C$. It
was proved the existence of generalized $p$-adic Gibbs measures in
both regimes. We obtained a $p$-adic dynamical system and
investigate its fixed points. In the first regime we establish that
the dynamical system has one attractive and two repelling fixed
points, which predicts the existence of a phase transition. In the
second regime the system has two attractive and one neutral fixed
points, which predicts the existence of a quasi phase transition.
Main results of the present paper are to verify (i.e. rigorously
prove) and confirm that the indicated predictions (via dynamical
systems point of view) are indeed true.
 These investigations show that there are some
similarities with the real case, for example, the existence of two
repelling fixed points implies the occurrence of the phase
transition. But there are also some differences. Namely, when the
dynamical system has two attractive fixed points, there occurs quasi
phase transition, unlike in real case, there is not such kind of
behavior. Finally, using such a method one can study other $p$-adic
models over trees.

\section*{Acknowledgement} The author thanks the MOE grant
ERGS13-024-0057.


\begin{thebibliography}{99}
\bibitem{F4} In the real case, when the
state space is compact, then the existence follows from the
compactness of the set of all probability measures (i.e. Prohorov's
Theorem). When the state space is non-compact, then there is a
Dobrushin's Theorem \cite{Dob1,Dob2} which gives a sufficient
condition for the existence of the Gibbs measure for a large class
of Hamiltonians.
\end{thebibliography}
\end{document}